\documentclass[smallextended]{svjour3}       
\smartqed  
\usepackage{graphicx}
\usepackage{color}
\usepackage{amsmath}
\usepackage{amssymb}
\usepackage{a4wide}
\usepackage{hyperref}
\usepackage{subfigure}
\usepackage{algorithm,algorithmic}

\newtheorem{thm}{Theorem}
\newtheorem{lem}[thm]{Lemma}
\newtheorem{prop}[thm]{Proposition}
\newtheorem{cor}[thm]{Corollary}
\newtheorem{defini}[thm]{Definition}
\newenvironment{defi}{\begin{defini}\rm}{\end{defini}}

\newenvironment{ex}{\begin{example}\rm}{\end{example}}
\newtheorem{rema}[thm]{Remark}
\newenvironment{rem}{\begin{rema}\rm}{\end{rema}}

\newcommand{\field}[1]{\mathbb{#1}}

\newcommand{\F}{\field{F}}

\newcommand{\cA}{{\mathcal A}}

\newcommand{\cC}{{\mathcal C}}

\newcommand{\cR}{{\mathcal R}}

\newcommand{\PG}{\mathcal{P}_{q}(n)}
\newcommand{\Gr}{\mathcal{G}_{q}(k,n)}
\newcommand{\G}{\mathcal{G}_{q}(k,n)}
\newcommand{\Uvs}{\mathcal{U}}
\newcommand{\Vvs}{\mathcal{V}}
\newcommand{\Cvs}{\mathcal{C}}
\newcommand{\Rvs}{\mathcal{R}}
\newcommand{\rs}{\mathrm{rs}}

\newcommand{\Gauss}[2]{
\left[\begin{array}
{c}#1\\#2\end{array}\right]_{q}
}

\begin{document}

\title{On the Geometry of Balls in the Grassmannian and List Decoding
  of Lifted Gabidulin Codes}

\author{Joachim Rosenthal \and Natalia Silberstein \and Anna-Lena
  Trautmann \thanks{J. Rosenthal and A.-L. Trautmann were partially
    supported by Swiss National Science Foundation Grant no.\ 138080.
    A.-L. Trautmann was partially supported by Forschungskredit of the
    University of Zurich, grant no.\ 57104103, and Swiss National Science Foundation Fellowship no.\ 147304.} \thanks{Parts of this
    work were presented at the International Workshop on Coding and
    Cryptography 2013 in Bergen, Norway, and appear in its proceedings
    \cite{tr13p}.}  }


\institute{J. Rosenthal \at
  Institute of Mathematics, University of Zurich, Switzerland\\
  \email{rosenthal@math.uzh.ch} 
\and A.-L. Trautmann \at
    Department of Electrical and Electronic Engineering, University of Melbourne, Australia\\
      \email{anna-lena.trautmann@unimelb.edu.au} 
  \and N. Silberstein \at Department of Computer Science,
  Technion --- Israel Institute of Technology, Haifa, Israel\\
  \email{natalys@cs.technion.ac.il}  }


\maketitle

\begin{abstract}
  The finite Grassmannian $\G$ is defined as the set of all
  $k$-dimensional subspaces of the ambient space $\F_{q}^{n}$.
  Subsets of the finite Grassmannian are called constant dimension
  codes and have recently found an application in random network
  coding. In this setting codewords from $\G$ are sent through a
  network channel and, since errors may occur during transmission, the
  received words can possibly lie in $\mathcal{G}_{q}(k',n)$, where
  $k'\neq k$.

  In this paper, we study the balls in $\G$ with center that is not
  necessarily in $\G$. We describe the balls with respect to two
  different metrics, namely the subspace and the injection metric.
  Moreover, we use two different techniques for describing these
  balls, one is the Pl\"ucker embedding of $\G$, and the second one is
  a rational parametrization of the matrix representation of the
  codewords.

  With these results, we consider the problem of list decoding a
  certain family of constant dimension codes, called lifted Gabidulin
  codes.  We describe a way of representing these codes by linear
  equations in either the matrix representation or a subset of the
  Pl\"ucker coordinates. The union of these equations
    and the linear and bilinear equations which arise from the
    description of the ball of a given radius provides an explicit
    description of the list of codewords with distance less than or
  equal to the given radius from the received word.
\end{abstract}

 \keywords{Grassmannian \and projective space \and subspace codes \and
network coding \and list decoding} \subclass{11T71,14G50}

\section{Introduction}

Let $\F_q$ be a finite field of size $q$ and let $k,n$ be two integers
satisfying $0 \le k \le n$.  The \textit{Grassmannian space}
(Grassmannian, in short), denoted by $\Gr$, is the set of all
$k$-dimensional subspaces of the vector space~\smash{$\F_q^n$}. Let
$\Uvs,\Vvs\subset \F_q^n$ be two different subspaces in $\Gr$.  The
\textit{subspace distance} is defined by
\begin{equation}
  \label{def_subspace_distance}
  d_S (\Uvs,\Vvs) =  \dim(\Uvs)+\dim(\Vvs) -2 \dim\bigl( \Uvs\, {\cap}\Vvs\bigr).
\end{equation}
A subset $\cC$ of $\Gr$ is called an $(n,M,d,k)_q$ \emph{constant
  dimension code} if it has size $M$ and if the minimum pairwise
subspace distance between any two different subspaces of $\cC$ is $d$.

Constant dimension codes gained a lot of interest due to the work by
K\"otter and Kschischang~\cite{KK} who showed that such codes are very
useful for error-correction in random network coding. They proved that
an $(n,M,d,k)_q$ code can correct any $\mu$ packet insertions
(which is equivalent to an increase of dimension by $\mu$ in the transmitted subspace) and $\epsilon$ packet deletions (which is equivalent to a decrease of dimension by $\epsilon$)
introduced anywhere in the network as long as $2\mu + 2\epsilon < d$. This
application has motivated extensive work in the
area~\cite{BoGa09,EtSi09,EV08,GaYa10,GNW12,GuXi12,KoKu08,MGR08,MaVa12,MaVa10,ro12b,SiEt09,SiEt10,Ska10,TMR10,TrRo10}.
In~\cite{KK} K\"otter and Kschischang gave a Singleton-like upper
bound on the size of such codes and presented a Reed-Solomon-like code
which asymptotically attains this bound.  Silva, K\"otter, and
Kschischang~\cite{si08j} showed how this construction can be described
in terms of lifted Gabidulin codes~\cite{ga85a}. The generalizations
of this construction and the decoding algorithms were presented
in~\cite{BoGa09,EtSi09,MGR08,ro12b,Ska10,TrRo10}. Another type of
construction (orbit codes) can be found in~\cite{EV08,KoKu08,TMR10}.

In this paper we focus on describing the balls of a given radius in
the Grassmannian around an arbitrary element of the respective
projective space. This is exactly what is needed to come up with list
decoding algorithms for constant dimension codes.  Then we focus on
list decoding of lifted Gabidulin codes.
For the classical Gabidulin codes it was recently shown by
Wachter-Zeh~\cite{WZ12} that, if the radius of the ball around a
received word is at least the Johnson radius, no polynomial-time
list decoding is possible, since the list size can be exponential.
Algebraic list decoding algorithms for folded Gabidulin codes were
discussed in~\cite{GNW12,MaVa12}. The constructions of subcodes of
(lifted) Gabidulin codes and their algebraic list decoding algorithms
were presented in~\cite{GW13,GuXi12,MaVa10,WZ13}.

One approach in this paper for list decoding codes in the Grassmannian
is to apply the techniques of Schubert calculus over finite fields,
i.e. to represent subspaces in the Grassmannian by their Pl\"ucker
coordinates.  It was proven in~\cite{ro12b} that a ball of a given
radius (with respect to the subspace distance) around a subspace can
be described by explicit linear equations in the Pl\"ucker embedding.
In this work we extend this result to the injection distance, which is
interesting for the case when a ball around a subspace of a different
dimension $k'\neq k$ is considered. Also, we describe a way of
representing a subset of the Pl\"ucker coordinates of lifted Gabidulin
codes as linear block codes, which results in additional linear
(parity-check) equations. The solutions of all these
  linear equations combined with the bilinear equations defining the
  Grassmannian in the Pl\"ucker embedding will constitute the
  resulting list of codewords. Another approach considered in this
paper is the description of the balls (for both the subspace and the
injection distance) around a subspace by bilinear equations from a
rational parametrization of the matrix representation of elements of
$\Gr$.

The paper is organized as follows. In Section~\ref{sec:preliminaries}
we review the Pl\"ucker embedding of the Grassmannian $\Gr$.
In Section~\ref{sec22} we describe the balls of radius $t$ around some
subspace of $\F_q^n$. We give the defining equations in Pl\"ucker
coordinates and also describe a rational parametrization which will
make the algorithmic computation for many list decoding problems
easier.  Section~\ref{sec:LG} contains the description of the lifted
Gabidulin codes as linear block codes.
Finally Section~\ref{sec:list_decoding} contains two list decoding
algorithms where we show how the set of equations describing a ball of
some radius and the equations describing the lifting of the Gabidulin
code can be computed.  Conclusions and problems for future research
are given in Section~\ref{sec:conclusions}.

\section{Preliminaries and Notations}
\label{sec:preliminaries}

We denote by $GL_{n}$ the general linear group over $\F_{q}$, by
$S_{n}$ the symmetric group on $n$ elements.
With $\mathbb P^{n}$ we denote the projective space of dimension $n$
over $\F_{q}$.

We represent some $\Uvs \in \G$ by the row space of a matrix $U\in
\F_{q}^{k\times n}$, where we use the notation $\rs(U)$ for the row
space of $U$. $GL_{n}$ acts on $\G$ as follows:
\begin{align*}
\G \times GL_{n} &\rightarrow \G\\
(\rs(U) , A) & \mapsto \rs(UA)   .
\end{align*}

Let $p(x)=\sum p_i x^i\in \F_q[x]$ be a monic and irreducible
polynomial of degree $\ell$, and $\alpha$ be a root of $p(x)$. Then it
holds that $\F_{q^{\ell}} \cong \F_q[\alpha]$. We denote the vector
space isomorphism between the extension field $\F_{q^{\ell}}$ and the
vector space $\F_q^{\ell}$ by
\begin{align*}
  \phi^{(\ell)} : \quad\F_{q^{\ell}} & \longrightarrow \F_q^{\ell} \\
  \sum_{i=0}^{\ell-1} \lambda_i \alpha^i &\longmapsto
  (\lambda_0,\dots,\lambda_{\ell-1}) .
\end{align*}
%
Moreover, we need the following notations:
The set of ordered multiindices of length $k$ with elements from $\{1,2,\dots,n\}$ is denoted by
\[\binom{[n]}{k} := \{(x_{1},\dots, x_{k}) \mid x_{i} \in
\{1,2,\dots,n\}, x_{1}<\dots<x_{k}\}, \] and for a matrix $A$ we
denote its $i$-th row by $A[i]$, its $i$-th column by $A_i$, and the
entry in the $i$-th row and the $j$-th column by $A_{i,j}$.

\begin{ex}
\[\binom{[4]}{2} = \{(1,2), (1,3), (1,4), (2,3), (2,4), (3,4)\}\]
\end{ex}

\begin{defi}
  The \emph{Bruhat order} on the set $\binom{[n]}{k}$, is defined as
  \[(i_{1},...,i_k) \preceq (j_{1},...,j_{k}) \iff i_{l} \leq j_{l}
  \quad\forall l \in \{1,\dots,k\}. \] The \emph{lexicographic order}
  is defined as,
$$
(i_{1},\dots, i_{k}) < (j_{1},\dots, j_{k})\iff \exists 0\leq N\leq k
: i_{m}=j_{m} \forall m \leq N \textnormal{ and } i_{N+1}< j_{N+1}.
$$
One notes that the Bruhat order is a partial order and the
lexicographic order is a total order on $\binom{[n]}{k}$.
\end{defi}

\begin{ex}According to the Bruhat order it holds that
  $(1,2,7)\preceq(2,3,7)$. But the fact that $(2,4,6)\not
  \preceq(2,3,7)$ does not imply that $(2,3,7)\prec(2,4,6)$. These two
  tuples are not comparable. In the lexicographic order it holds that
  $(1,2,7) < (2,3,7)$ and $(2,3,7)<(2,4,6)$.
\end{ex}

We denote by $\PG$ the set of all subspaces of $\F_q^n$, i.e.,
$$
\PG:=\bigcup_{0\leq k\leq n}\G.
$$

\begin{defi}
  Let $\Uvs, \Vvs\in \PG$ be two subspaces.  The \emph{subspace
    distance} is defined as
  \[
  d_S(\Uvs, \Vvs) = \dim(\Uvs) + \dim(\Vvs) - 2\dim(\Uvs\cap \Vvs)
  \]
  and the \emph{injection distance} is definded as
  \[
  d_I(\Uvs, \Vvs) = \max(\dim(\Uvs), \dim(\Vvs)) - \dim(\Uvs\cap \Vvs)
  .
  \]
\end{defi}

Clearly both distance functions describe a metric in the usual way.
One also notes that for $\Uvs, \Vvs \in \G$ it holds that
$d_S(\Uvs,\Vvs)= 2d_I(\Uvs,\Vvs)$. Moreover, for constant dimension
codes a unique subspace distance decoder is equivalent to a unique
injection distance decoder \cite{tr13phd}.
For list decoding we will derive a similar relation between
the two metrics in Proposition \ref{prop:relation}.

\begin{defi}
  We define the balls in $\G$ with subspace radius $\tau$ around an
  arbitrary element $\Rvs \in \PG$ as
  \[B^k_{S,\tau}(\Rvs) := \{\Vvs \in \G \mid d_S(\Rvs, \Uvs) \leq
  \tau\}.\] Analogously we define the balls in $\G$ with injection
  radius $t$ around an arbitrary element $\Rvs \in \PG$ as
  \[{B}^k_{I,t}(\Rvs) := \{\Vvs \in \G \mid d_I(\Rvs, \Uvs)
  \leq t\}.\]
\end{defi}


The Pl\"ucker embedding of the Grassmannian is a useful tool when
studying $\G$.  The basic idea of using the Pl\"ucker embedding for
list decoding of subspace codes was already stated in
\cite{ro12b,tr12a}. We will now recall the main definitions and
theorems from those works. The proofs of the results can also be found
in there. For more information or a more general formulation of the
Pl\"ucker embedding and its applications the interested reader is
referred to \cite{ho52}.

\begin{rem}
  The condition $d_S(\Rvs, \Uvs) \leq\tau$ (respectively $d_I(\Rvs, \Uvs)
  \leq t$)
  translates into the condition that a subspace $\Uvs$ should
  intersect the received space $\Rvs$ in at least a certain dimension.
  Geometrically this describes a so-called ``Schubert condition'' and
  actually both $B^k_{S,\tau}(\Rvs)$ and ${B}^k_{I,t}(\Rvs)$ have
  the structure of a so called ``Schubert variety''. Readers familiar
  with Schubert calculus as described in~\cite{ho52} will readily
  recognize this and it will not come as a surprise that the Pl\"ucker
  equations which describe the balls will turn out to be linear.
  In order to keep the paper as self contained as possible
we will derive in this paper the relevant equations.
\end{rem}

Let $U\in \F_{q}^{k\times n}$ such that its row space $\rs (U)$
describes the subspace $\Uvs \in \Gr$. $M_{i_1,\dots,i_k}(U)$ denotes
the minor (i.e. the determinant of the submatrix) of $U$ given by the
columns $i_1,\dots,i_k$.  The Grassmannian $\Gr$ can be embedded into the
projective space $\mathbb{P}^{\binom{n}{k}-1}$ of dimension $\binom{n}{k} -1$ over $\F_q$ using the Pl\"ucker embedding:
\begin{align*}
  \varphi : \Gr  &\longrightarrow \mathbb{P}^{\binom{n}{k}-1} \\
  \rs (U) &\longmapsto [M_{1,...,k}(U) : M_{1,...,k-1,k+1}(U) :\ldots
  : M_{n-k+1,...,n}(U)].
\end{align*}
%
%
The $k\times k$ minors $M_{i_1,\ldots ,i_k}(U)$ of the matrix $U$ are
called the \emph{Pl\"ucker coordinates} of the subspace~$\Uvs$. By
convention, we order the minors lexicographically by the column
indices.

The image of this embedding describes indeed a variety and the
defining equations of the image are given by the so called
\emph{shuffle relations} (see e.g.~\cite{kl72,pr82}), which are
multilinear equations of monomial degree $2$ in terms of the Pl\"ucker
coordinates:
\begin{prop}[\cite{kl72,pr82}]           \label{shuffle}
  Consider $x:=[x_{1,\dots,k}:\dots:x_{n-k+1,\dots,n}]\in
  \mathbb{P}^{\binom{n}{k}-1}$. Then there exists a subspace $\Uvs \in
  \Gr$ such that $\varphi(\Uvs)=x$ if and only if
  \[\sum_{j\in\{i_{1},\dots,i_{k+1}\}} {\mathrm{sgn}({\sigma}_{j})}
  x_{i_{1},\dots,i_{k+1}\setminus j} x_{j,i_{k+2},\dots,i_{2k}} = 0\]
  $ \forall (i_1,\dots,i_{k+1}) \in \binom{[n]}{k+1},
  (i_{k+2},\dots,i_{2k}) \in \binom{[n]}{k-1}$, where
  $\mathrm{sgn}(\sigma_{j})$ denotes the sign of the permutation such
  that
$$
\sigma_{i_{\ell}}(i_{1},\dots,i_{k+1})=(i_{\ell},i_{1},\dots,i_{\ell-1},i_{\ell+1},\dots,i_{k+1})
.
$$
\end{prop}

Then one can easily derive an upper bound on the number of shuffle
equations.

\begin{lem}\label{lem:num}
  There are at most $\binom{n}{k+1} \binom{n}{k-1}$ different
  (non-trivial) shuffle relations defining $\Gr$ in the Pl\"ucker
  embedding.
\end{lem}

\begin{ex}\label{ex:shuffle}
  $\mathcal{G}_q(2,4)$ is described by a single relation:
  $$x_{12}x_{34}-x_{13}x_{24}+x_{14}x_{23} = 0.$$
\end{ex}

\section{Balls in the Grassmannian $\G$}\label{sec22}


\subsection{Description by linear equations in the Pl\"ucker
  embedding}\label{sec:linear}

It is known that the equations defining the balls inside $\G$ around
an element from $\G$ are easily determined in the following special
case:

\begin{prop}[\cite{ho52,ro12b}]\label{prop5}
  Define $\Uvs_{0}:=\rs [\begin{array}{cc}I_{k} &0_{k\times n-k}
  \end{array}]$. Then for $t\leq k-1$
  \begin{align*}
    B_{S,2t}^k(\Uvs_{0}) = \{\Vvs=\rs (V) \in & \Gr \mid
    M_{i_1,...,i_{k}}(V) = 0 \\ &\forall (i_{1},...,i_{k}) \not
    \preceq (t+1,\dots,k,n-t+1,...,n) \}.
  \end{align*}
  Note that for $t=k$ it holds that $B_{S,2k}^k(\Uvs)= \Gr$ for any
  $\Uvs\in\Gr$.
\end{prop}

We now want to state a generalization of this fact, where the center
of the ball can have a different dimension than $k$. For this we first
need the following lemma.

\begin{lem}\label{lemevenodd}
  Let $\Uvs,\Vvs \in \PG$ with $\dim(\Uvs)= k$ and $\dim(\Vvs)=k'$.
  \begin{enumerate}
  \item Then $d_S(\Uvs, \Vvs) $ is odd if and only if exactly one of
    $k$ and $k'$ is odd. Equivalently $d_S(\Uvs, \Vvs)$ is even if
    and only if both $k$ and $k'$ are odd or if both are even.
  \item It holds that $k-k'+d_S(\Uvs,\Vvs)$ and $k'-k+d_S(\Uvs,\Vvs)$
    are always even numbers.
  \end{enumerate}

\end{lem}
\begin{proof}
  It holds that $d_S(\Uvs, \Vvs) = k+k' -2 \dim(\Uvs\cap \Vvs) $, i.e.
  it is odd if and only if $k+k'$ is odd. This directly implies the
  first statement. The second statement follows since $k-k'$ and
  $k'-k$ are odd if and only if exaclty one of $k$ and $k'$ is odd, as
  well.
   \qed
\end{proof}

We can now state the generalization of Proposition \ref{prop5} for the
subspace distance.
\begin{thm}
  \label{thm:gen_ball_eq}
  Let $\Uvs_{0}^{k'}:=\rs [\begin{array}{cc}I_{k'} &0_{k'\times
      n-k'} \end{array}]$. Then
      for
  $|k-k'| \leq\tau< \min(k'+k, 2n-(k'+k))$, s.t.
  $\frac{k+k'-\tau}{2}\in \mathbb{Z}$ (which we can assume because of
  Lemma \ref{lemevenodd})
  \begin{align*}
    B_{S,\tau}^{k}(\Uvs_{0}^{k'}) = \Big\{\Vvs=\rs (V) \in \Gr &\mid
    M_{i_1,...,i_{k}}(V) = 0 \;\forall\; (i_{1},...,i_{k}) \not
    \preceq \\& \left(\frac{k'-k+\tau}{2}+1,\dots ,k',
      n-\frac{k-k'+\tau}{2}+1,\dots ,n \right) \Big\}.
  \end{align*}
\end{thm}

\begin{proof}
  We want to find all $\Vvs=\rs(V) \in \G$,
  such that
  \[d_S(\Uvs_0^{k'}, \Vvs) \leq \tau\]
  \[\iff \dim (\Uvs_0^{k'} \cap \Vvs) \geq \frac{k+k'-\tau}{2}\]
  i.e. at least $\frac{k+k'-\tau}{2}$ many linearly independent
  elements of $\Vvs$ have to be in $\Uvs_0^{k'}$. Thus, we can choose
  a matrix representation of the form
  \[ V=\left[\begin{array}{c|c} * & 0_{(\frac{k+k'-\tau}{2}) \times
        (n-k')} \\\hline * & *\end{array} \right] .\] Each $k\times
  k$-submatrix of $V$ is then of the form
  \[ M=\left[\begin{array}{c|c} M_{1} & 0_{(\frac{k+k'-\tau}{2})
        \times x} \\\hline M_2 & M_{3}\end{array} \right] \] where
  $0\leq x\leq k$ is the number of columns taken from the $n-k'$ right
  most columns of $V$ and $M_1$ is a $\frac{k+k'-\tau}{2}\times(k-x)$
  matrix.  Since $\mathrm{rank}(M) \leq \mathrm{rank}(M_{1})+
  \mathrm{rank}([M_{2} M_3]) \leq (k-x)+\frac{k-k'+\tau}{2}=
  k-(x-\frac{k-k'+\tau}{2})$
 it follows that all minors of $V$ that
  contain at least $x= \frac{k-k'+\tau}{2}+1$ of the $n-k'$ rightmost
  columns are zero. At the same time this is also a sufficient
  condition, since the $*$-blocks of $V$ can be filled with anything
  (such that the whole matrix has rank $k$) and the row space will
  always be in the ball.
  Since the monomials are ordered, the condition that at least
  $\frac{k-k'+\tau}{2}+1$ many coordinates of $(i_1,\dots,i_k)$ are in
  $\{k'+1,\dots,n\}$ is equivalent to the condition that
  \[i_{\ell} \geq k'+1 \textnormal{ for some } {\ell} \in
  \left\{1,\dots, \frac{k+k'-\tau}{2}\right\} \]
  which is in turn equivalent to
  \[(i_{1},\dots ,i_{k}) \not \preceq
  \left(k'-\frac{k+k'-\tau}{2}+1,\dots ,k',
    n-\frac{k-k'+\tau}{2}+1,\dots ,n \right) \]
  \[\iff (i_{1},\dots ,i_{k}) \not \preceq
  \left(\frac{k'-k+\tau}{2}+1,\dots ,k', n-\frac{k-k'+\tau}{2}+1,\dots
    ,n \right).
  \]
  \qed
\end{proof}

In analogy, we can also state the generalization of Proposition
\ref{prop5} for the injection distance:

\begin{thm}
  \label{thm:gen_ball_eq_inj}
  Define $\Uvs_{0}^{k'}$ as before. Then for $|k'-k| \leq t <
  \min(\max(k',k), n-k+1)$
  \begin{align*}
    {B}_{I,t}^{k}(\Uvs_{0}^{k'}) = \Big\{\Vvs=\rs (V) &\in \Gr
    \mid M_{i_1,...,i_{k}}(V) = 0 \;\forall\; (i_{1},...,i_{k}) \not
    \preceq \\& \left(k' -\max (k',k)+t+1,\dots ,k',
      n-k+\max(k',k)-t+1,\dots ,n \right) \Big\}.
  \end{align*}
\end{thm}

\begin{proof}
  We want to find all $\Vvs=\rs(V) \in \G$, such that
  \[d_I(\Uvs_0^{k'}, \Vvs) \leq t\]
  \[\iff \dim(\Uvs_0^{k'} \cap \Vvs) \geq \max (k',k)-t ,
  \]
  i.e. at least $\max (k',k)-t$ many linearly independent elements of
  $\Vvs$ have to be in $\Uvs_0^{k'}$. Thus, we can choose a matrix
  representation of the form
  \[ V=\left[\begin{array}{c|c} * & 0_{( \max (k',k)-t) \times (n-k')}
      \\\hline * & *\end{array} \right] .\] Analogously to the proof
  of Theorem \ref{thm:gen_ball_eq} this is equivalent to the statement
  that all minors containing at least $ \min (0,k-k')+t+1$ of the
  $n-k'$ rightmost columns are zero, which is in turn equivalent to
  \[(i_{1},\dots ,i_{k}) \not \preceq \left(k'-\max (k',k)+t+1,\dots
    ,k', n-k+\max(k',k)-t+1,\dots ,n \right).
  \]
   \qed
\end{proof}

The following proposition shows that the conditions on $\tau$ and $t$
in the previous theorems make sense.

\begin{prop}
  Let $\Uvs \in \mathcal{G}_{q}(k',n)$.
  \begin{enumerate}
  \item For $\tau=\min(k'+k, 2n-(k'+k))$ it holds that
    $B_{S,\tau}^k(\Uvs)= \Gr$.
  \item For $t=\min(\max(k',k), n-\max(k',k))$ it holds that
    ${B}_{I,t}^k(\Uvs)= \Gr$.
  \item For $\tau < |k'-k|$ it holds that
    $B_{S,\tau}^k(\Uvs)={B}_{I,\tau}^k(\Uvs)= \emptyset$.
  \end{enumerate}
\end{prop}
\begin{proof}
  Let $\Vvs \in \G$.
  \begin{enumerate}
  \item Let $\tau=k'+k$. Then $d_{S}(\Uvs, \Vvs)=k'+k \iff \dim(\Uvs
    \cap \Vvs)= 0 \iff B_{S,k+k'}^k(\Uvs)= \Gr$.

    Let $\tau=2n-(k+k')$. Since it is known that $d_{S}(\Uvs,
    \Vvs)=d_{S}(\Uvs^{\perp} , \Vvs^{\perp} )$ it holds that
    $$
    B_{S,2n-(k+k')}^k(\Uvs)=
    (B_{S,(n-k)+(n-k')}^{n-k}(\Uvs^{\perp}))^{\perp}=
    \mathcal{G}_{q}(n-k,k)^{\perp} = \Gr.
    $$
  \item Let $t=\max(k',k)$. Then $d_{I}(\Uvs, \Vvs)=\max(k',k) \iff
    \dim(\Uvs \cap \Vvs)= 0 \iff {B}_{I,\max(k',k)}^k(\Uvs)=
    \Gr$.

    Let $t=n-\min(k',k)=\max(n-k', n-k)$. Since it is known that
    $d_{I}(\Uvs, \Vvs)=d_{I}(\Uvs^{\perp} , \Vvs^{\perp} )$ it holds
    that $B_{I,n-\min(k',k)}^k(\Uvs)=
    (B_{I,\max(n-k', n-k)}^{n-k}(\Uvs^{\perp}))^{\perp}=
    \mathcal{G}_{q}(n-k,k)^{\perp} = \Gr$.
  \item Moreover, $d_{S}(\Uvs, \Vvs)< |k'-k| \iff d_{I}(\Uvs, \Vvs)<
    |k'-k| \iff \dim(\Uvs \cap \Vvs)> \min(k,k') \iff
    B_{S,k+k'}^k(\Uvs)= {B}_{I,k+k'}^k(\Uvs)= \emptyset$.
  \end{enumerate}
   \qed
\end{proof}

\begin{rem}
  The linear equations described in
  Theorems \ref{thm:gen_ball_eq} and \ref{thm:gen_ball_eq_inj} together
  with the shuffle relations described in Proposition~\ref{shuffle}
  show that the balls ${B}_{S,t}^{k}(\Uvs_{0}^{k'})$ as well as
  the balls ${B}_{I,t}^{k}(\Uvs_{0}^{k'})$ are sub-varieties of the
  Grassmann variety $\Gr$.
\end{rem}

\begin{rem}
  In Theorems \ref{thm:gen_ball_eq} and \ref{thm:gen_ball_eq_inj}, for
  $t=k'-k$ (if $k'\geq k$) the formula for the balls becomes
  \begin{align*}
    {B}_{S,t}^{k}(\Uvs_{0}^{k'}) ={B}_{I,t}^{k}(\Uvs_{0}^{k'}) =\\
    \Big\{\Vvs=\rs (V) &\in \Gr \mid M_{i_1,...,i_{k}}(V) = 0
    \;\forall\; (i_{1},...,i_{k}) \not \preceq \left(k' -k+1,\dots ,k'
    \right) \Big\}.
  \end{align*}
\end{rem}

\begin{ex}\label{ex11}
  \begin{enumerate}
  \item Consider $\mathcal{G}_{q}(2,6)$ and $\Uvs_{0}^{3}=\rs\left(
      \begin{array}{cccccc}1&0&0&0&0&0 \\ 0&1&0&0&0&0 \\ 0&0&1&0&0&0
      \end{array} \right)$. Then $k=2, k'=3$ and
    \[B_{S,3}^{2} (\Uvs_{0}^{3}) = {B}_{S,2}^{2} (\Uvs_{0}^{3}) =
    \{\Vvs=\rs (V) \in \mathcal G_{q}(2,6) \mid M_{i_{1},i_{2}}(V) = 0
    \forall (i_{1},i_{2}) \not \preceq (3, 6) \}.
    \]
  \item Consider $\mathcal{G}_{q}(3,6)$ and $\Uvs_{0}^{2}=\rs\left(
      \begin{array}{cccccc}1&0&0&0&0&0 \\ 0&1&0&0&0&0 \end{array}
    \right)$. Then $k=3, k'=2$ and
    \[B_{S,3}^{3} (\Uvs_{0}^{2}) = {B}_{I,2}^{3} (\Uvs_{0}^{2}) =
    \{\Vvs=\rs (V) \in \mathcal G_{q}(3,6) \mid M_{i_{1},i_{2},
      i_{3}}(V) = 0 \forall (i_{1},i_{2}, i_{3}) \not \preceq (2, 5,
    6) \}.
    \]
  \end{enumerate}
\end{ex}

We can find a relation for the balls of the two different metrics as
follows.

\begin{prop}\label{prop:relation}
   Let $\Uvs \in \mathcal{G}_{q}(k',n)$ and $\Vvs \in \G$. Then
  \[d_{S}(\Uvs, \Vvs) = 2d_{I}(\Uvs, \Vvs) + k+k'-2\max(k',k) \] and
  \[{B}_{I,t}(\Uvs) = B_{S,2t+k+k'-2\max(k',k)}(\Uvs) .\]
\end{prop}

\begin{proof}
First, it holds that
 \begin{align*}2d_{I}(\Uvs, \Vvs) + k+k'-2\max(k',k)=2\max(k',k)-2\dim(\Uvs\cap \Vvs) +k+k'-2\max(k',k)\\
 =k+k'-2\dim(\Uvs\cap \Vvs)=d_{S}(\Uvs, \Vvs)
   \end{align*}
 Second, it holds that
  \begin{align*}
    \Vvs \in B_{S,2t+k+k'-2\max(k',k)}(\Uvs)
    \iff d_{S}(\Uvs, \Vvs) \leq 2t +k+k'-2\max(k',k)\\
    \iff k+k'- 2\dim(\Uvs \cap \Vvs) \leq 2t+k+k'-2\max(k',k)
    \iff \dim(\Uvs \cap \Vvs) \geq {\max(k',k)-t}\\
    \iff \max(k',k)- \dim(\Uvs \cap \Vvs) \leq t \iff d_{I}(\Uvs,
    \Vvs) \leq t \iff\Vvs \in {B}_{I,t}(\Uvs) .
  \end{align*}
  \qed
\end{proof}

With the knowledge of $B_{S,\tau}^{k}(\Uvs_0^{k'})$ we can also express
$B_{S,\tau}^{k}(\Uvs)$ for any $\Uvs \in \mathcal{G}_q(k',n) $. To do so we need the following result.

\begin{lem}\label{lem:generalballs}
  For any $\Uvs \in \mathcal{G}_q(k',n) $ there exists an $A\in GL_n$
  such that $\Uvs_0^{k'} A= \Uvs$.
   Moreover,
  \[B_{S,\tau}^{k}(\Uvs_0^{k'} A) = B_{S,\tau}^k(\Uvs_0^{k'}) A .\] The
  same holds for the injection distance, i.e.
  \[{B}_{I,\tau}^{k}(\Uvs_0^{k'} A) =
  {B}_{I,\tau}^k(\Uvs_0^{k'}) A .\]
\end{lem}
\begin{proof}
Both statements follow from the fact that $\dim(\Uvs_0^{k'} A \cap \Vvs) = \dim(\Uvs_0^{k'}  \cap \Vvs A^{-1})$, since this directly implies that $d_S(\Uvs_0^{k'} A , \Vvs)\leq \tau \iff d_S(\Uvs_0^{k'}  , \Vvs A^{-1})$ and $d_I(\Uvs_0^{k'} A , \Vvs)\leq \tau \iff d_I(\Uvs_0^{k'}  , \Vvs A^{-1})$.
\qed
\end{proof}

\begin{rem}
\label{rm:A}
Note that one can easily find $A\in GL_n$
  such that $\Uvs_0^{k'} A= \Uvs$ as follows: Let the upper $k'$ rows of $A$ be equal to the reduced row echelon form of $\Uvs$ and fill the lower rows with unit vectors such that the respective ones and the pivots of the upper rows are all in different columns. This implies that $A$ is invertible and that $\Uvs_0^{k'} A= \Uvs$. For an algorithmic description of constructing such an $A$ see \cite{ro12b,tr13phd}.
  \end{rem}

The following results are generalizations of results from
\cite{ro12b}.
For simplifying the computations we define
$\bar{\varphi}$ 
on $GL_n$, where we denote by
$A_{j_1,\ldots,j_k}[{i_{1},\dots, i_k}]$ the submatrix of
$A$ that consists of the rows $i_{1}, \dots, i_{k}$ and columns $j_1,\ldots,j_k$:

\begin{align*}
  \bar{\varphi} : GL_n &\longrightarrow GL_{\binom{n}{k}} \\
  A & \longmapsto \left(\begin{array}{cccccc}
      \det A_{1,\dots, k}[1, \dots, k] & \dots & \det A_{n-k+1 ,\dots, n}[{1,\dots, k}]\\
      \vdots & & \vdots \\
      \det A_{1, \dots, k}[{n-k+1,\dots, n}] & \dots & \det
      A_{n-k+1,\dots, n}[n-k+1 ,\dots, n]
    \end{array}
  \right)
\end{align*}


\begin{lem}[\cite{ro12b}]\label{lem:2}
  Let $\Uvs \in \Gr $ and $A\in GL_{n}$. It holds that
  \[\varphi(\Uvs A) = \varphi(\Uvs) \bar{\varphi}(A).\]
\end{lem}

Since it holds for any $k$, we can use this lemma
to describe a ball around a subspace
 of arbitrary dimension.

\begin{cor}\label{thm5}
  Let $\Uvs= \Uvs_{0}^{k'}A \in \mathcal{G}_q(k',n) $. Then
  \begin{align*}
    B_{S,\tau}^k(\Uvs) =B_{S,\tau}^k(\Uvs_{0}^{k'} A) =\Big\{&\Vvs=\rs (V)
    \in \Gr \mid M_{i_{1},\dots,i_{k}}(V)\bar \varphi(A^{-1}) = 0\;
    \forall (i_{1},\dots,i_{k})\not \preceq \\&
    \left(\frac{k'-k+\tau}{2}+1,\dots ,k',
      n-\frac{k-k'+\tau}{2}+1,\dots ,n \right)\Big\},
    \\
    {B}_{I,t}^k(\Uvs) =B_{I,t}^k(\Uvs_{0}^{k'} A)
    =\Big\{&\Vvs=\rs (V) \in \Gr \mid M_{i_{1},\dots,i_{k}}(V)\bar
    \varphi(A^{-1}) = 0\; \forall (i_{1},\dots,i_{k})\not \preceq \\&
    \left(k'-\max (k',k)+t+1,\dots ,k', n-k+\max(k',k)-t+1,\dots ,n
    \right) \Big\}.
  \end{align*}
\end{cor}


In the following we calculate the number of equations which define a
ball of a given radius.
\begin{lem}\label{lemnumber}
  The maximum number of linear Pl\"ucker equations defining a ball
  $B_{S,\tau}^k(\mathcal{U}_0^{k'})$ (respectively
  ${B}_{I,t}^k(\mathcal{U}_0^{k'})$) is equal to the maximum number of
  equations defining $B_{S,\tau}^{k}(\Uvs)$ (respectively
  ${B}_{I,t}^k(\mathcal{U})$) for any $\Uvs \in
  \mathcal{G}_q(k',n)$.
\end{lem}
\begin{proof}
  Follows directly from Corollary \ref{thm5}.  \qed \end{proof}

We can hence count the maximum number of linear equations needed to describe
the ball inside the Grassmannian.

\begin{lem}
  Let $\Uvs \in \G$.  An upper bound on the number of linear equations
  needed to describe $B_{S,\tau}^k(\mathcal{U})$ is
  \[\theta_{S} := \sum_{l=0}^{\frac{k+k'-\tau}{2}-1}
  \binom{n-k'}{k-\ell} \binom{k'}{\ell} .
  \]
  An upper bound on the number of linear equations needed to describe
  ${B}_{I,t}^k(\mathcal{U})$ is
  \[\theta_{I} := \sum_{l=0}^{\max(k',k)-t-1}\binom{n-k'}{k-\ell}
  \binom{k'}{\ell} .
  \]
\end{lem}
\begin{proof}
  Follows from Lemma \ref{lemnumber} and Theorems
  \ref{thm:gen_ball_eq} and \ref{thm:gen_ball_eq_inj}.
   \qed
\end{proof}

Note that all these equations defining a ball (in subspace or
injection metric) are linearly independent. This can be seen by the description of the balls around $\Uvs_{0}^{k'}$, since the equations are of the form $M_{i_1,\dots,i_k}(V) =0 $ for different minors functioning as the variables and are thus linearly independent. As the equations describing the balls around arbitrary elements can be found by linear transformations, these equations will also be linearly independent.


\subsection{Description by rational
  parametrization}\label{sec:rational}

One can also use a rational parametrization to describe the balls
around $\Uvs_{0}^{k'}$ in the Grassmannian as follows.

\begin{prop}\label{prop:rational}
  Define $\Uvs_{0}^{k'}$ as previously and $\nu:= \frac{k-k'+\tau}{2},
  \omega:=\min(0,k-k')+t$. Then
  \begin{align*}
    B_{S,\tau}^{k}(\Uvs_{0}^{k'}) =& \Big\{\Vvs=\rs [V_{1} \; V_{2}] \in
    \Gr &\mid V_{1}\in \F_{q}^{k\times k'}, \exists X\in
    \F_{q}^{k\times \nu}, Y\in \F_{q}^{\nu \times (n-k')}: V_{2} = XY
    \Big\}
  \end{align*}
  and
  \begin{align*}
    {B}_{I,t}^{k}(\Uvs_{0}^{k'}) =& \Big\{\Vvs=\rs [V_{1} \;
    V_{2}] \in \Gr &\mid V_{1}\in \F_{q}^{k\times k'}, \exists X\in
    \F_{q}^{k\times \omega}, Y\in \F_{q}^{\omega \times (n-k')}: V_{2}
    = XY \Big\}.
  \end{align*}
\end{prop}

\begin{proof}
  We want to find all $\Vvs=\rs[V_{1} \; V_{2}] \in \G$ such that
$$ d_{S}(\Uvs_{0}^{k'}, \Vvs) \leq \tau \iff  \mathrm{rank}
\left[\begin{array}{ccc} I_{k'} &&0_{k'\times
      (n-k')} \\
    V_{1}&&V_{2}
  \end{array}\right]
\leq \frac{k+k'+\tau}{2} $$
 $$
 \iff k'+\mathrm{rank} (V_{2}) \leq \frac{k+k'+\tau}{2} \iff
 \mathrm{rank} (V_{2}) \leq \nu. $$
 The last statement is equivalent to the fact that there exists $X\in
 \F_{q}^{k\times \nu}, Y\in \F_{q}^{\nu \times (n-k')}$ such that
 $V_{2}=XY$. The proof for the injection distance
 is analogous.
 \qed
\end{proof}

\begin{rem}
As the proof shows for the description of the balls $B_{S,\tau}^{k}$ and
${B}_{I,t}^{k}$ it is crucial to describe all $k\times m$
matrices $V_2$ whose rank is at most $\nu$. The set of all $k\times m$
matrices of rank at most $\nu$ is sometimes called a {\em
  determinantal variety $D^{k\times m}_\nu$}. These varieties are known to be rational,
this means there is a birational isomorphism from a Zariski
open subset of this variety to an open subset of a vector space.
To make this concrete in our setting identify the set of all
$k\times \nu$ matrices $X=\left[ X_1\atop X_2\right]$,
where the top part $X_1$ is an invertible matrix, with an
open subset of the  vector space $\F^{k\nu}$. Similarly identify
the set of $\nu \times m$ matrices having the form $Y=[ I_\nu \ Y_2]$
with the vector space  $\F^{\nu\times(m-\nu)}$. Then

\begin{eqnarray*}
f : \F_{q}^{k\nu}\times \F_{q}^{\nu\times (m-\nu)} &\longrightarrow&
D^{k\times m}_\nu\\
(X,Y)&\longmapsto & XY
\end{eqnarray*}
defines a birational isomorphism. The map in particular provides
a ``rational parametrization'' of the variety $D^{k\times m}_\nu$
and in particular the dimension of $D^{k\times m}_\nu$ is equal
to $k\nu+m\nu -\nu^2$. Note that not all points of the variety
$D^{k\times m}_\nu$ are parametrized but the description avoids
dealing with many equations describing the vanishing of the
$(\nu+1)\times (\nu+1)$ minors.
\end{rem}

In analogy to Section \ref{sec:linear} we can also describe the balls
around arbitrary elements in $\G$ in a similar manner.

\begin{thm}\label{thm:rational}
  Let $\Rvs = \rs [\; R_{1} \; R_2 \; ] \in \mathcal{G}_{q}(k', n)$
  such that $R_1 \in \F_q^{k'\times k'}, R_2 \in \F_q^{k'\times
    (n-k')}$. Moreover, let $\nu$ and $\omega$ be as before. Then
  there exists
  \[A = \left(\begin{array}{c|c} R_{1} & R_{2}\\\hline R_{3} & R_{4}
    \end{array}\right) \quad \in \mathrm{GL}_{n}\]
  such that $\Rvs = \rs[ I_{k'} \; 0_{k' \times (n-k')}] A$. It holds
  that
  \begin{align*}
    B_{S,\tau}^{k}(\Rvs) =& \Big\{\Vvs=\rs \left(
   [ V_{1}+V_{2}]\left[
\begin{array}{cc} R_{1} & R_{2}\\ R_{3} & R_{4}
    \end{array} \right]\right) \in \Gr \mid \\
    &\hspace{5cm}V_{1}\in \F_{q}^{k\times k'}, \exists X\in
    \F_{q}^{k\times \nu}, Y\in \F_{q}^{\nu \times (n-k')}: V_{2} = XY
    \Big\}
  \end{align*}
  and
  \begin{align*}
    {B}_{I,t}^{k}(\Rvs) =& \Big\{\Vvs=\rs [\;
    V_{1}R_{1}+V_{2}R_{3} \hspace{0.4cm} V_1 R_2 + V_{2}R_{4}\; ] \in
    \Gr \mid \\ &\hspace{5cm}V_{1}\in \F_{q}^{k\times k'}, \exists
    X\in \F_{q}^{k\times \omega}, Y\in \F_{q}^{\omega \times (n-k')}:
    V_{2} = XY \Big\}.
  \end{align*}
\end{thm}
\begin{proof}
  We know from Lemma \ref{lem:generalballs} that $B_{S,\tau}^{k} ( \Rvs)
  = B_{S,\tau}^{k}(\Uvs_0^{k'}) A$. Together with Proposition
  \ref{prop:rational} it follows that
  \begin{align*}
    B_{S,\tau}^{k}(\Rvs) =& \Big\{\Vvs=\rs [\; V_{1} \; V_{2}\;] A \in
   \Gr  \mid  V_{1}\in \F_{q}^{k\times k'}, \exists X\in \F_{q}^{k\times \nu},
   Y\in \F_{q}^{\nu \times (n-k')}: V_{2} = XY  \Big\}\\
    =& \Big\{\Vvs=\rs [\; V_{1}R_{1}+V_{2}R_{3} \hspace{0.4cm} V_1 R_2
    + V_{2}R_{4}\; ] \in \Gr \mid \\ & \hspace{5cm} V_{1}\in
    \F_{q}^{k\times k'}, \exists X\in \F_{q}^{k\times \nu}, Y\in
    \F_{q}^{\nu \times (n-k')}: V_{2} = XY \Big\} .
  \end{align*}
  The proof for the injection distance is analogous.  \qed
\end{proof}

\begin{cor}
  In the setting of Theorem \ref{thm:rational}, if $R_{1}$ has full
  rank, one can choose a matrix representation of the form $\Rvs = [
  I_{k'} \; \tilde R_{2} ]$. Then the formulas are simplified to
  \begin{multline*}
    B_{S,\tau}^{k}(\Rvs) =\\
   \Big\{\Vvs=\rs [\; V_{1} \hspace{0.4cm} V_1
    \tilde R_2 + V_{2}\;] \in \Gr \mid V_{1}\in \F_{q}^{k\times k'},
    \exists X\in \F_{q}^{k\times \nu}, Y\in \F_{q}^{\nu \times
      (n-k')}: V_{2} = XY \Big\}
  \end{multline*}
  and
  \begin{multline*}
    {B}_{I,t}^{k}(\Rvs) =\\
 \Big\{\Vvs=\rs [\; V_{1}
    \hspace{0.4cm} V_1 \tilde R_2 + V_{2}\;] \in \Gr \mid V_{1}\in
    \F_{q}^{k\times k'}, \exists X\in \F_{q}^{k\times \omega}, Y\in
    \F_{q}^{\omega \times (n-k')}: V_{2} = XY \Big\}.
  \end{multline*}
\end{cor}
\begin{proof}
  Because of the shape of $\Rvs$ we can choose
  \[A= \left(\begin{array}{cc} I_{k'} & \tilde R_2 \\
      0 & I_{n-k'}
    \end{array}\right)\]
  such that $\rs [\; I_{k'}\; 0 \;] A = \rs [\; R_{1} \; R_2 \;]$.
  Then the statement follows from Theorem \ref{thm:rational}.  \qed
\end{proof}

\begin{ex}\label{ex:rational1}
  Consider $\mathcal{G}_q(2,4)$ and let
  $\Rvs = \rs\left(\begin{array}{cccc}1&0&0 & 1 \\
      0&1 &1 & 1 \end{array}\right)$.  Then $k=k'=2$ and
  \begin{multline*}
    B_{S,2}^{2}(\Rvs) =\\
     \Big\{\rs \left[ V_{1} \hspace{0.2cm} V_1
      \left( \begin{array}{cc}0&1\\1&1 \end{array}
      \right) + V_{2}\;\right] \in  \mathcal{G}_{q}(2,4)  \mid V_1\in\F_q^{2\times 2},
    \quad V_{2} = (X_1 X_2)^T (Y_1 Y_2), \quad X_1,X_2,Y_1,Y_2 \in \F_q  \Big\}  \\
    = \Big\{\rs \left( \begin{array}{cccc} a & b & b+X_1 Y_1 &
        a+b+X_1 Y_2 \\ c&d & d+X_2 Y_1 & c+d+X_2 Y_2
      \end{array}\right) \in \mathcal{G}_{q}(2,4)
    \mid a,b,c,d,X_1,X_2,Y_1,Y_2 \in \F_q
    \Big\} .
  \end{multline*}
\end{ex}

In the description of the balls from Theorem \ref{thm:rational},
define $\bar V := V_1 R_1 + V_{2} R_{3}$ (i.e. the left part of the
elements in the ball) and $\tilde V := V_1 R_2 + V_{2} R_{4}$ (i.e.
the right part of the elements in the ball). Then for a given $R_2$
and some $V_1 \in \F_q^{k\times k}$ one gets a set of bilinear
equations of the form
\[ \bar V_{ij} = \sum_{\ell=1}^{k'} (V_1)_{i\ell} (R_1)_{\ell j} +
\sum_{\ell=1}^{n-k'}\sum_{m=1}^{\nu} X_{im} Y_{m\ell} (R_{3})_{\ell j}
\]
\[ \tilde V_{ij} = \sum_{\ell=1}^{k'} (V_1)_{i\ell} (R_2)_{\ell j} +
\sum_{\ell=1}^{n-k'}\sum_{m=1}^{\nu} X_{im} Y_{m\ell} (R_{4})_{\ell j}
\] for the subspace distance and of the form
\[ \bar V_{ij} = \sum_{\ell=1}^{k'} (V_1)_{i\ell} (R_1)_{\ell j} +
\sum_{\ell=1}^{n-k'}\sum_{m=1}^{\omega} X_{im} Y_{m\ell} (R_{3})_{\ell
  j} \]
\[ \tilde V_{ij} = \sum_{\ell=1}^{k'} (V_1)_{i\ell} (R_2)_{\ell j} +
\sum_{\ell=1}^{n-k'}\sum_{m=1}^{\omega} X_{im} Y_{m\ell} (R_{4})_{\ell
  j} \] for the injection distance. From this we can determine the
degree and the number of variables of this system of equations:

\begin{lem}\label{lem:rational}
  For a given $V_{1}$, the description of the balls from Theorem
  \ref{thm:rational} results in a system of bilinear equations in
  $kk'+(n-k'+k)\nu$ unknowns for the subspace distance, respectively
  $kk'+(n-k'+k)\omega$ unknowns for the injection distance, given by
  $V_{1}, X$ and $Y$.
\end{lem}

To sum up, we know how to describe the balls, in both the subspace and
the injection metric, in $\G$ with a given radius around an element of
$\PG$ with either linear equations in the Pl\"ucker embedding or
bilinear equations in the matrix coordinates. In the following
sections we will show how this can be used to establish list decoding
algorithms for lifted Gabidulin codes.



\section{Lifted Gabidulin Codes}\label{sec:LG}

For two $k \times \ell$ matrices $A$ and $B$ over $\F_q$ the {\it rank
  distance} is defined by
$$
d_R (A,B) := \text{rank}(A-B)~.
$$
A $[k \times \ell,\varrho,\delta]$ {\it rank-metric code} $C$ is a
linear subspace with dimension $\varrho$ of $\F_q^{k \times \ell}$, in
which each two distinct codewords $A$ and $B$ have distance $d_R (A,B)
\geq \delta$. For a $[k \times \ell,\varrho,\delta]$ rank-metric code
$C$ it was proven in~\cite{de78,ga85a,ro91} that
\begin{equation}
  \label{eq:MRD} \varrho \leq
  \text{min}\{k(\ell-\delta+1),\ell(k-\delta+1)\}~.
\end{equation}
Codes which attain this bound are called {\it maximum rank distance}
codes (or MRD codes in short).

An important family of MRD linear codes was presented by
Gabidulin~\cite{ga85a}.  These codes can be seen as the analogs of
Reed-Solomon codes for the rank metric.  From now on let $k\leq \ell$.
A codeword $A$ in a $[k \times \ell, \varrho , \delta]$ rank-metric
code $C$ can be represented by a vector $c_{A}=(c_1 , c_2 , \ldots ,
c_{k})$, where $c_i={\phi^{(\ell)}}^{-1}(A[i]) \in \F_{q^{\ell}}$.
Let $g_i\in \F_{q^{\ell}}$, 
$1\leq i\leq k$, be linearly independent over $\F_q$.  Then the
generator matrix $G$
of a $[k \times \ell,\varrho,\delta]$ Gabidulin MRD code is given by
$$
G=\left(\begin{array}{cccc}
    g_{1} & g_{2} & \ldots & g_{k}\\
    g_{1}^{[1]} & g_{2}^{[1]} & \ldots & g_{k}^{[1]}\\
    g_{1}^{[2]} & g_{2}^{[2]} & \ldots & g_{k}^{[2]}\\
    \vdots & \vdots & \vdots & \vdots\\
    g_{1}^{[k-\delta]} & g_{2}^{[k-\delta]} & \ldots &
    g_{k}^{[k-\delta]}\end{array}\right),
$$
where $\varrho=\ell (k-\delta+1)$, and $[i]=q^{i}$~\cite{ga85a}.

Let $A$ be a $k \times \ell$ matrix over $\F_q$ and let $I_k$ be the
$k \times k$ identity matrix. The matrix $[ I_k ~ A ]$ can be viewed
as a generator matrix of a $k$-dimensional subspace of
$\F_q^{k+\ell}$. This subspace is called the \emph{lifting} of
$A$~\cite{si08j}.

When the codewords of a rank-metric code $C$ are lifted to
$k$-dimensional subspaces, the result is a constant dimension code
$\cC$. If $C$ is a Gabidulin MRD code then $\cC$ is called a
\emph{lifted Gabidulin
  code}.

\begin{thm}[\cite{si08j}]
  \label{trm:param lifted MRD}
  Let $k$, $n$ be positive integers such that ${k \leq n-k}$.  If $C$
  is a $[k \times (n-k), (n-k)(k-\delta +1),\delta ]$ Gabidulin MRD
  code then $\cC$ is an $(n,q^{(n-k)(k-\delta+1)},2\delta, k)_{q}$
  constant dimension code.
\end{thm}


We will now show that the row expansion of a Gabidulin code forms a linear
block code.
Let $C$ be an $[k\times\ell, \ell(k-\delta+1),\delta)]$ Gabidulin MRD
code over $\F_q$, $k\leq \ell$. We denote by $C^L$ the linear block
code of length $k\ell$ over $\F_q$, such that every codeword $c^A$ of
$C^L$ is obtained from a codeword $A\in C$ by taking the entries of $A$,
 row by row, from  bottom to  top, left to right (w.l.o.g.).

\begin{thm}
\label{thm:linear_from_Gab}
 The code $C^L$ is  a $[k\ell,\ell(k-\delta+1),\geq\delta]$ linear block code over $\F_q$  in the Hamming metric.
\end{thm}
\begin{proof}
 The linearity of $C^L$ directly follows from the linearity of
 $C$. The length of $C^L$ is the number of entries in a  codeword of
 $C$, and $C$ and $C^L$ have the same cardinality.
 Since the rank of each non-zero $A\in C$ is greater or equal to $\delta$,
  also the number of non-zero entries of $A$ has to be greater or equal
   to $\delta$, hence the minimum Hamming distance $d_{min}$ of $C^L$ satisfies $d_{min}\geq \delta$.
    \qed
\end{proof}

We denote by $H^L$ a parity-check matrix of $C^L$.
\vspace{.3cm}



We will now show that also a subset of the Pl\"ucker coordinates of a
lifted Gabidulin code is a linear block code over $F_q$.

As before, let $C$ be an $[k\times(n-k), (n-k)(k-\delta+1,\delta)]$
Gabidulin MRD code over $\F_q$. Then by Theorem~\ref{trm:param lifted
  MRD} its lifting is a code $\cC$ of size $q^{(n-k)(k-\delta+1)}$ in
the Grassmannian $\Gr$. Let
$$
x^{\cA}=[x^{\cA}_{1\ldots k}:\ldots : x^{\cA}_{n-k+1\ldots n}] \in
\mathbb{P}^{\binom{n}{k}-1}
$$
be a vector which represents the Pl\"ucker coordinates of a subspace
$\cA\in \Gr$. If $x^{\cA}$ is normalized (i.e. the first non-zero entry is equal to one),
 then $x^{\cA}_{1\ldots k}=1$ for any $\cA\in \cC$.

Let $[k]=\{1,2,\ldots,k\}$, and let $\underline{i}=\{i_1,i_2,\ldots,i_k\}$
 be a set of indices such that $|\underline{i}\cap[k]|=k-1$.
 Let $t\in \underline{i}$,
such that $t>k$, and $s=[k]\setminus\underline{i}$.

\begin{lem}\label{lem7}
  Consider $A\in C$ and $\cA=\rs[\;I_k \; A\;]$. If $x^{\cA}$ is
  normalized, then $x^{\cA}_{\underline{i}}=(-1)^{k-s}A_{s,t-k}$.
\end{lem}
\begin{proof}
  It holds that $x^{\cA}$ is normalized if its entries are the minors
  of the reduced row echelon form of $\cA$, which is $[\;I_k \; A\;]$.
  Because of the identity matrix in the first $k$ columns, the
  statement follows directly from the definition of the Pl\"ucker
  coordinates.  \qed
\end{proof}

Note, that we have to worry about the normalization since $x^{\cA}$ is
projective. In the following we will always assume that any element
from $\mathbb{P}^{\binom{n}{k}-1}$ is normalized.

Similarly to Theorem \ref{thm:linear_from_Gab}, with Lemma \ref{lem7}
one can easily show, that a subset of the Pl\"ucker coordinates of a
lifted Gabidulin code forms a linear code over $\F_q$:

\begin{thm} \label{thm9} The restriction of the set of Pl\"ucker
  coordinates of an $(n,q^{(n-k)(k-\delta+1)},2\delta, k)_{q}$ lifted
  Gabidulin code $\cC$ to the set $\{\underline{i}:|\underline{i}|=k,
  |\underline{i}\cap[k]|=k-1\}$ forms a linear code $C^p$ over $\F_q$
  of length $k(n-k)$, dimension $(n-k)(k-\delta+1)$ and minimum
  Hamming distance $d_{min}\geq \delta$.
\end{thm}

\begin{rem}
When $\F_q=\F_2$, then $C^p$ is equivalent to $C^L$.
\end{rem}

\vspace{.2cm}

We denote by $H^p$ a parity-check matrix of $C^p$.

\begin{ex}\label{ex:Gab}
  Let $\alpha \in \F_{2^2}$ be a primitive element, fulfilling
  $\alpha^2=\alpha+1$.  Let $C$ be the $[2\times2, 2, \delta=2]$
  Gabidulin MRD code over $\F_2$ defined by the generator matrix $
G=(\alpha\;1)$.
In this example we want to consider the lifting of $C=\{(b\alpha,b):
b\in \F_{2^2}\}$.  The codewords of $C$, their representation as
$2\times2$ matrices, their lifting to $\mathcal{G}_2(2,4)$ and the
respective Pl\"ucker coordinates are given in the following table.

\begin{center}
\begin{tabular}{|c|c|c|c|}
  \hline
  vector representation  & matrix representation & lifting & Pl\"ucker coordinates \\
  \hline
  \hline
$(0,0)$&
      $ \left( \begin{array}{cc}
                 0 & 0 \\
                 0 & 0 \\
               \end{array}\right)$ &
                $\rs\left(  \begin{array}{cccc}
          1 & 0 & 0 & 0 \\
          0 & 1 & 0 & 0 \\  \end{array} \right) $ &
           $[1:0:0:0:0:0]$ \\\hline
$(\alpha,1)$ &
       $\left( \begin{array}{cc}
                 0 & 1\\
                 1 & 0 \\
               \end{array} \right)$ &
                $\rs\left(  \begin{array}{cccc}
          1 & 0 & 0 & 1 \\
          0 & 1 & 1 & 0 \\\end{array}\right)$ &
           $[1:1:0:0:1:1]$ \\  \hline
$(\alpha^2,\alpha)$&
       $\left(      \begin{array}{cc}
                 1 & 1 \\
                 0 & 1 \\ \end{array} \right)$ &
                   $\rs\left( \begin{array}{cccc}
          1 & 0 & 1 & 1 \\
          0 & 1 & 0 & 1 \\\end{array}    \right) $ &
           $[1:0:1:1:1:1]$ \\\hline
$(1,\alpha^2)$&
      $ \left(\begin{array}{cc}
                 1 & 0 \\
                 1 & 1 \\
               \end{array} \right) $&
                $\rs\left(\begin{array}{cccc}
          1 & 0 & 1 & 0 \\
          0 & 1 & 1 & 1 \\
        \end{array}\right)$ & $[1:1:1:1:0:1]$ \\\hline
\end{tabular}
\end{center}

\vspace{0.5cm}

In this example, $C^p=\{(0000),(1001),(0111),(1110)\}$. This is a
$[4,2,2]$ linear code in the Hamming space. Its parity-check matrix is
$$
H^p=\left(
        \begin{array}{cccc}
          1 & 0 & 1 & 1 \\
          0 & 1 & 1 & 0 \\
        \end{array}
      \right).
$$
In other words, a Pl\"ucker coordinate vector
$[x_{12}:x_{13}:x_{14}:x_{23}:x_{24}:x_{34}]$ of a vector space from
$\mathcal{G}_{2}(2,4)$ represents a codeword of the lifted Gabidulin
code from above if and only if $x_{12}=1$, $x_{14}+x_{23}=0$, and
$x_{13}+x_{23}+x_{24}=0$.

\end{ex}


\section{List Decoding of Lifted Gabidulin Codes}
 \label{sec:list_decoding}

We now have all the machinery needed to describe two list decoding
algorithms for lifted Gabidulin codes, one in the Pl\"ucker
coordinates and another one in the matrix entries. We will describe
everything in this section using the subspace distance. The
translation of these results to the injection metric is then
straight-forward. In this section we will describe both list decoding
algorithms and  give a bound on the list size for
 lifted Gabidulin codes.

 \subsection{List decoding in the Pl\"ucker embedding}\label{sec:41}

 Consider a lifted Gabidulin code $\mathcal{C} \subseteq \G$ and
 denote its corresponding $[k(n-k),(n-k)(k-\delta+1)]$-linear block
 code over $\F_q$ 
 by $C^{p}$. The corresponding parity check matrix is denoted by
 $H^{p}$. Let $\mathcal{R} =\rs(R) \in \Gr$ be the received word.

 We showed in Section \ref{sec:LG} how a subset of the Pl\"ucker
 coordinates of a lifted Gabidulin code forms a linear block code that
 is defined through the parity check matrix $H^{p}$. Since we want to
 describe a list decoding algorithm inside the whole set of Pl\"ucker
 coordinates, we define an extension of $H^{p}$ as follows:
 \[\bar H^{p} = \left[\begin{array}{cccc}
     0_{(\delta-1)(n-k)\times 1} & H^{p} &  0_{(\delta-1)(n-k)\times
       \ell}
   \end{array}\right ]\]
 where $\ell = \binom{n}{k} - k(n-k) -1$. Then $[x_{1\dots k}:\ldots :
 x_{n-k+1\dots n}]\bar {H^{p}}^{T} = 0$ gives rise to the same
 equations as $[x_{i_{1}}:\ldots:x_{i_{k(n-k)}}]{H^{p}}^{T}=0$, for
 $i_{1},\dots,i_{k(n-k)} \in \underline i$. For simplicity we will
 write $\bar x$ for $[x_{1\dots k}:\ldots : x_{n-k+1\dots n}]$ in the
 following.

\begin{lem}
  The linear equations $\bar x\bar {H^{p}}^{T} = 0$ together with the
  normalization condition $x_{1,...,k}=1$ and the shuffle relations
  described in Proposition~\ref{shuffle} describe the lifted Gabidulin
  code $\Cvs$ in terms of its Pl\"ucker coordinates.
\end{lem}
\begin{rem}
  Using the language of algebraic geometry one can also say that
  $\Cvs$ has the structure of a quasi-projective sub-variety of the
  Grassmann variety $\Gr$.
\end{rem}

The list decoding problem up to the decoding radius $\tau$ requires
the explicit description of the intersection of the varieties
$$
L^\tau_\Cvs(\mathcal{R}):=\Cvs \cap B_{S,\tau}^k(\mathcal{R}),
$$
which we will call the {\em list variety} of the received subspace
$\mathcal{R}$.  The following algorithm provides an explicit
computation of the equations describing $L^\tau_\Cvs(\mathcal{R})$.

\begin{algorithm}[H]
  \caption{}\label{alg:1}
{\normalsize
Input: received word $\mathcal{R} \in \PG$, decoding radius $\tau$

\begin{enumerate}
\item Find the (linear) equations defining $B_{S,\tau}^k(\mathcal{R})$
in the Pl\"ucker coordinates, as explained in Section~\ref{sec22}.
\item Solve the system of (linear) equations, that arises from
$\bar x \bar {H^{p}}^{T}=0$, together with the equations of $B_{S,\tau}^k(\mathcal{R})$,
the (bilinear) shuffle relations and the equation $x_{1,\dots, k}=1$ (describing the lifting).
\end{enumerate}
Output: the solutions $\bar x = [x_{1\dots k}:\ldots : x_{n-k+1\dots
  n}]$ of this system of equations }
\end{algorithm}

Note that there exist many algorithms to solve bilinear equations that
one can use in Step 2. of the algorithm, see e.g.
\cite{Court00,KiSh99,ThWo10}. In this paper we will consider the
relinearization algorithm from \cite{KiSh99}.

\begin{thm}
  Algorithm \ref{alg:1} outputs the complete list $L$ of codewords (in
  Pl\"ucker coordinate representation), such that for each element
  $\bar x \in L$, $d_S(\varphi^{-1}(\bar x),\mathcal{R})\leq \tau$.
\end{thm}

\begin{proof}
  The solution set to the shuffle relations is exactly $\varphi(\Gr)$,
  i.e. all the elements of $\mathbb{P}^{\binom{n}{k}-1}$ that are
  Pl\"ucker coordinates of a $k$-dimensional vector space in
  $\F_{q}^{n}$. The subset of this set with the condition $x_{1,\dots,
    k}=1$ is exactly the set of Pl\"ucker coordinates of elements in
  $\G$ whose reduced row echelon form has $I_{k}$ as the left-most
  columns. Intersecting this with the solution set of the equations
  given by $H^{p}$ achieves the Pl\"ucker coordinates of the lifted
  code $\Cvs$.  The intersection with $B_{S,\tau}^k(\mathcal{R})$ is
  then given by the additional equations from Step 1 in the algorithm.
  Thus the solution set to the whole system of equation is the
  Pl\"ucker equations of $\Cvs \cap B_{S,\tau}^k(\mathcal{R})$.  \qed
\end{proof}

\begin{ex}
  \label{ex:Alg1}
  We consider the $(4,4,4,2)_2$ lifted Gabidulin code from Example
  \ref{ex:Gab}. Note, that for a received space of dimension $2$ it is
  not possible to decode always to a unique closest codeword.
  \begin{enumerate}
  \item Assume we received
    \[\Rvs_1 = \rs\left( \begin{array}{cccc} 1&0&1&0 \\ 0&0&0&1
      \end{array}\right) .\]
    We would like to correct one error. We first find the equations
    for the ball of subspace radius~$2$:
    \[B_{S,2}^2(\Uvs_{0}^2) = \{\Vvs=\rs(V) \in \mathcal{G}_{2}(2,4)
    \mid M_{3,4}(V) = 0 \}\] We
    construct 
    \[A^{-1} _1= \left( \begin{array}{cccc} 1&0&0&1 \\ 0&0&1&0 \\
        0&0&0&1 \\ 0&1&0&0\end{array}\right) \] 
       such that $\Rvs_1 A_1^{-1} = \Uvs_{0}^2$ (see
    Remark~\ref{rm:A}) and compute the last column of
    $\bar{\varphi}(A^{-1}_1)$:
    \[[1:0:0:1:0:0]^{T}.\] Thus, by Corollary~\ref{thm5} we get that
    \[B_{S,2}^2(\Rvs_1) = \{\Vvs=\rs(V) \in \mathcal{G}_{2}(2,4) \mid
    M_{1,4}(V)+M_{2,3}(V) = 0 \}.\] Then combining with the parity
    check equations from Example~\ref{ex:Gab} we obtain the following
    system of linear equations to solve
    \begin{align*}
      x_{13}+x_{14}+x_{24}&=0\\
      x_{14}+x_{23}&=0\\
      x_{12}+x_{23}&=0\\
      x_{12} &= 1
    \end{align*}
    where the first two equations arise from $\bar H^{p}$, the third
    from $B_{S,2}^2(\Rvs_1) $ and the last one represents the identity
    submatrix. This system has the two solutions $(1,1,1,1,0)$ and
    $(1,0,1,1,1)$ for $(x_{12},x_{13},x_{14},x_{23},x_{24})$. Since we
    used all the equations defining the ball in the system of
    equations, we know that the two codewords corresponding to these
    two solutions (i.e. the third and fourth in Example \ref{ex:Gab})
    are the ones with distance $2$ from the received space, and we do
    not have to solve $x_{34}$ at all. The corresponding codewords are
    \[\rs\left( \begin{array}{cccc}
        1 & 0 & 1 & 0 \\
        0 & 1 & 1 & 1 \\ \end{array} \right) , \rs\left(
      \begin{array}{cccc}
        1 & 0 & 1 & 1 \\
        0 & 1 & 0 & 1 \\ \end{array} \right) .\]

  \item

    Now assume we received
    \[\Rvs_2 = \rs\left( \begin{array}{cccc} 1&0&0&1 \\ 0&1&1&1
      \end{array}\right) .\]
    As previously, we
    construct 
    \[A^{-1} _2= \left( \begin{array}{cccc} 1&0&0&1 \\ 0&1&1&1 \\
        0&0&1&0 \\ 0&0&0&1\end{array}\right) \] (see
    Remark~\ref{rm:A}) and compute the last column of
    $\bar{\varphi}(A^{-1}_2)$:
    \[[1:1:0:1:1:1]^{T}.\] Thus, by Corollary~\ref{thm5} we get that
    \[B_{S,2}^2(\Rvs_1) = \{\Vvs=\rs(V) \in \mathcal{G}_{2}(2,4) \mid
    M_{1,2}(V)+M_{1,3}(V)+M_{2,3}(V)+M_{2,4}(V)+M_{3,4}(V) = 0 \}.\]
    Then combining with the parity check equations from
    Example~\ref{ex:Gab} and the shuffle relation from
    Example~\ref{ex:shuffle} we obtain the following system of linear
    and bilinear equations:
    \begin{align*}
      x_{13}+x_{14}+x_{24}&=0\\
      x_{14}+x_{23}&=0\\
      x_{12}+x_{13}+x_{23}+x_{24}+x_{34}&=0\\
      x_{12}x_{34}+x_{13}x_{24}+x_{14}x_{23}&=0\\
      x_{12} &= 1
    \end{align*}

    We rewrite these equations in terms of the variables
    $x_{13},x_{14},x_{23},x_{24}$ which correspond to a lifted
    Gabidulin code as follows.
    \begin{align*}
      x_{13}+x_{14}+x_{24}&=0\\
      x_{14}+x_{23}&=0\\
      x_{1,3}+x_{2,3}+x_{2,4}+x_{13}x_{24}+x_{14}x_{23}&=1\\
    \end{align*}

    This system has three solutions $(1,0,0,1)$, $(0,1,1,1)$, and
    $(1,1,1,0)$ for $(x_{13},x_{14},x_{23},x_{24})$.  The
    corresponding codewords are
    \[\rs\left( \begin{array}{cccc}
        1 & 0 & 0 & 1 \\
        0 & 1 & 1 & 0 \\  \end{array} \right),
    \rs\left(  \begin{array}{cccc}
        1 & 0 & 1 & 1 \\
        0 & 1 & 0 & 1 \\  \end{array} \right),
    \rs\left(  \begin{array}{cccc}
        1 & 0 & 1 & 0 \\
        0 & 1 & 1 & 1 \\  \end{array} \right) .\]
  \end{enumerate}
\end{ex}

\begin{rem}
  In the previous example, for the same code and two 
  received words of the same dimension, in one case we needed the
  bilinear shuffle relations whereas in the other case we could
  completely list decode without taking the shuffle relations into
  account. Thus, the actual shape of the received space can make a
  difference for the complexity of the decoding algorithm.
\end{rem}

The complexity of Algorithm \ref{alg:1} is dominated by solving the
system of $\theta_S+1+(\delta-1)(n-k)+\binom{n}{k-1}\binom{n}{k+1}$
linear and bilinear equations in $\binom{n}{k}$ variables.


\begin{thm}
  Using the relinearization algorithm from \cite{KiSh99}, the
  complexity of Algorithm \ref{alg:1} is polynomial in $n$ and
  exponential in $k$.
\end{thm}

\begin{proof}
  We can use the relinearization algorithm of \cite{KiSh99} to solve
  the system of linear and bilinear equations in Algorithm
  \ref{alg:1}.  This algorithm is polynomial in the number of
  variables if the number of equations is at least the square of the
  number of variables, which is satisfied in our case, since
  $\binom{n}{k-1}\binom{n}{k+1} \approx \binom{n}{k}^2$.
  With the approximation $\binom{n}{k} \approx n^k$, the statement
  follows.  \qed
\end{proof}

Note that it is not easy to determine the actual complexity of the
relinearization algorithm as described in \cite{KiSh99}. The paper
states that the number of arithmetic operations is a polynomial
$\psi(N)$ where $N$ is the number of variables involved. For our
situation that would translate that the number of arithmetic
operations is $O\left(\psi(n^{2k})\right)$ once $k$ is small in
comparison to $n$.



\subsection{List decoding with the rational parametrization}

We can use the description of the balls from Section
\ref{sec:rational} with the additional constraints from the
description of the lifted Gabidulin codes, i.e. the first $k\times
k$-block is the identity and the rightmost $n-k$ columns fulfill the
parity check equations from the linear code description.

\begin{algorithm}[H]
  \caption{}\label{alg:2}
{\normalsize
Input: received word $\mathcal{R} \in \PG$, decoding radius $\tau$
\begin{enumerate}
\item Find the (bilinear) equations defining $B_{S,\tau}^k(\mathcal{R})$
 in the rational parametrization, as explained in Section~\ref{sec:rational}.
\item Solve the system of (linear) equations, that arises from $H^{L}$,
 together with the equations of $B_{S,\tau}^k(\mathcal{R})$ and the
equations corresponding to the first block of the codewords being equal to the identity. (In the notation of Theorem \ref{thm:rational} the variables are given by the matrices $V_1, X$ and $Y$.)
\item For each solution from 2.\ compute $U=[\; V_{1}R_{1}+V_{2}R_{3} \hspace{0.4cm} V_1 R_2
    + V_{2}R_{4}\; ] $ (in the notation of Theorem \ref{thm:rational}).
\end{enumerate}
Output: matrices $U$, whose row spaces are the codewords in $B_{S,\tau}^k(\mathcal{R})\cap \cC.$
}
\end{algorithm}

\begin{ex}
  \begin{enumerate}
  \item Consider $\mathcal{G}_2(2,4)$ and the code from Example
    \ref{ex:Gab}. Let the received word be
    $\Rvs = \rs\left(\begin{array}{cccc}1&0&0 & 1 \\
        0&1 &1 & 1 \end{array}\right)$, as in Example~\ref{ex:Alg1}.2
    Then we know from Example \ref{ex:rational1} that
    \begin{align*}
      B_{S,2}^{2}(\Rvs) =& \Big\{\rs \left( \begin{array}{cccc} a & b
          & b+X_1 Y_1 & a+b+X_1 Y_2 \\ c&d & d+X_2 Y_1 & c+d+X_2 Y_2
        \end{array}\right) \in \Gr \mid a,b,c,d,X_1,X_2,Y_1,Y_2 \in \F_q
      \Big\} .
    \end{align*}
    Since we want to find only codewords of the lifted Gabidulin code
    in the ball, we can set $a=d=1$ and $b=c=0$. We label the entries
    of the third and fourth column from bottom left to top right by
    $(v_{1}, \dots, v_{4})=(1+X_2Y_1,1+X_2Y_2,X_1Y_1,1+X_1Y_2)$. With
    the parity-check equations from the code $C^L$ (which is the same
    as $C^P$ in this case)
    \[v_{1}+v_3+v_{4}=0 \quad v_{2} + v_{3}= 0 \] we get the following
    system of equations:
    \begin{align*}
      (1 + X_{2}Y_{1}) + (1 + X_{1} Y_{1})+(1+X_1Y_2) = 0& \quad (1 +
      X_{2}Y_{2}) + X_{1} Y_{1} = 0  ,
    \end{align*}
    which has the following solutions:
    \begin{align*}
      (X_{1},X_{2},Y_{1}, Y_{2}) \in \{(0,1,0,1),(1,0,1,1),(1,1,1,0)\}
      .
    \end{align*}
    These correspond to the codewords (remember that one has to add
    $a=d=1$ in some coordinates)
    \[\rs\left( \begin{array}{cccc}
        1 & 0 & 0 & 1 \\
        0 & 1 & 1 & 0 \\  \end{array} \right),
    \rs\left(  \begin{array}{cccc}
        1 & 0 & 1 & 0 \\
        0 & 1 & 1 & 1 \\  \end{array} \right),
    \rs \left(  \begin{array}{cccc}
        1 & 0 & 1 & 1 \\
        0 & 1 & 0 & 1 \\  \end{array} \right).
    \]
%
  \item Consider the same setting as before but let the received word
    be $\Rvs = \rs\left(\begin{array}{cccc}0&1&1 & 1
      \end{array}\right)$.
    Then we can choose
    \[A= \left( \begin{array}{cccc} 0&1&1&1 \\ 1&0&0&0 \\ 0&0&1&0 \\
        0&0&0&1 \end{array} \right) ,\] 
        such that $(1000) A = (0111)$
    and get by Theorem~\ref{thm:rational}
    \begin{align*}
      B_{S,1}^{2}(\Rvs) =& \Big\{\rs \left( \begin{array}{cccc}
          X_{1}Y_{1} & a & a+X_1 Y_2 & a+X_1 Y_3 \\ X_{2}Y_{1}&b &
          b+X_2 Y_2 & b+X_2 Y_3 \end{array}\right) \in \Gr \mid
      a,b,X_1,X_2,X_{3},Y_1,Y_2,Y_{3} \in \F_q \Big\} .
    \end{align*}
    Since we want to find only codewords in the ball, we can set $a=0$
    and $b=1$ and get the constraints $X_{1}Y_{1}=1, X_{2}Y_{1}=0$.
    With the equations from the code we get the following system of
    equations:
    \begin{align*}
      X_{1}Y_{1}=1& \quad X_{2}Y_{1}=0\\
      X_{1}Y_{2} + X_{2} Y_{3} = 1& \quad X_{1}Y_{2} + X_{1} Y_{3} +
      X_{2} Y_{2} = 1 ,
    \end{align*}
    which has the unique solution
    \begin{align*}
      (X_{1},X_{2},Y_{1}, Y_{2}, Y_{3}) = (1,0,1,1,0) .
    \end{align*}
    This corresponds to the codeword (remember that one has to add
    $b=1$ in some coordinates)
    \[\rs\left( \begin{array}{cccc}
        1 & 0 & 1 & 0 \\
        0 & 1 & 1 & 1 \\  \end{array} \right).
    \]
  \end{enumerate}
\end{ex}

We can do the following complexity analysis for Algorithm \ref{alg:2}.

\begin{thm}
  Using the relinearization algorithm from \cite{KiSh99}, Algorithm
  \ref{alg:2} has a computational complexity that is polynomial in $n$
  and exponential in $k$ (if the list size is small enough).
\end{thm}
\begin{proof}
  We know from Lemma \ref{lem:rational} that the system of bilinear
  equations to be solved in the algorithm has
  $kk'+n\frac{k-k'+\tau}{2}$ variables, which we can approximate by
  $kk'+n\frac{\tau}{2}$, if we assume $k\approx k'$. Moreover, it has
  at most $k^{2}$ equations for the identity part and
  $(n-k)(\delta-1)$ equations for the linear Gabidulin code
  description (see Theorem \ref{thm9}).  Since $\delta \leq k$ we can
  upper bound the number of equations by $(n-k)k + k^{2}= nk$.  We now
  use the relinearization algorithm for solving the system of
  equations. In this algorithm, either the second linearization has a
  unique solution
  or it has a solution space of dimension that is polynomial in $k$.
  Then we have to do the last steps for finding the solutions of the
  original variables for any of the elements of this solution space.

  Since the whole relinearization algorithm is polynomial if there is
  only one solution to the second linearization (see \cite{KiSh99}),
  our algorithm will have at most a complexity that is exponential in
  $k$.
  \qed
\end{proof}

Note, that if one is interested to get a list of codewords
  within a certain distance of the received word explicitly, then the
  efficiency of a decoding algorithm depends (at least) on the size of the list.
  In other words, if there is a list of exponential size, no
  polynomial time  algorithm can exist which explicitly outputs the total list. From an application point of view the list size is also important, since usually one wants to have a small list size to have sensible list decoding. This is due to the fact that one wants to choose one codeword of the output list after decoding to be the most likely sent codeword.
  Hence, we investigate the worst possible list sizes in the
  following.



We will derive a lower bound on the worst case list size for lifted
Gabidulin codes in analogy to the theorems and proofs of \cite{WZ12},
where these bounds were derived for classical Gabidulin codes.
We denote such a worst case list size, i.e.\ the maximum number of
codewords of an $(n,q^{(n-k)(k-\delta+1)},2\delta,k)_q$ lifted
Gabidulin code $\cC$ in a ball of a subspace radius $\tau$ around any
received word, by $L_S(\tau,n,k,\delta,q)$, and for injection radius
$t$ by $L_I(t,n,k,\delta,q)$.

\begin{thm}
  \label{trm:list_lower_bound}
  Lower bounds on the list sizes $L_S(\tau,n,k,\delta,q)$ and
  $L_I(t,n,k,\delta,q)$, for $t,\tau/2<\delta\leq k\leq n/2$, are
  given by
  \[L_S(\tau,n,k,\delta,q)\geq \frac{\Gauss{k}{\lfloor
      \tau/2\rfloor}}{q^{(n-k)(\delta-\lfloor \tau/2\rfloor-1)}}
  \quad \textnormal{ and } \quad 
  L_I(t,n,k,\delta,q)\geq \frac{\Gauss{k}{t}}{q^{(n-k)(\delta-t-1)}},
  \]
  where $\Gauss{a}{b}=\prod_{i=0}^{b-1}\frac{q^{a-i}-1}{q^{b-i}-1}$ is
  the $q$-ary Gaussian coefficient.
\end{thm}
%
%

  \begin{proof}
    First, we observe that to present a lower bound on
    $L(\tau,n,k,\delta,q)$ ($L_I(t,n,k,\delta,q)$) it is sufficient to
    consider the list size for a given received subspace, i.e. an
    existence of one such received subspace with a given list size
    provides the desired lower bound.  We consider a received word
    $\cR$ of the same dimension $k$.  Let $\cR:=\rs [\;I_k\; A_1\;]$
    for some $A_1\in \F_q^{k \times (n-k)}$. Then
$$d_S(\cR, \rs [\;I_k\; A\;]) =2 d_R(A_1, A), $$
for any $A\in \F_q^{k \times(n-k)}$ (see e.g.\ \cite{si08j}), and
hence the distance between $\cR$ and any codeword -- and more
generally any element from $\Gr$ -- is an even number. Thus, if $\tau$
is even, then $B_{\tau +1}^k(\cR) = B_{\tau}^k(\cR)$ and hence
$L_S(\tau +1,n,k,\delta,q)=L_S(\tau,n,k,\delta,q)$.
%
%
Furthermore, if $\tau$ is even, $\rs [\;I_k\; A\;]$ is in the ball
around $\cR$ of subspace radius $\tau$ if and only if $A$ is in the
ball around $A_1$ of rank radius $\tau/2$.
It follows that the lower bound of the list size of classical
Gabidulin codes for rank radius $\lfloor \tau/2\rfloor$ is also a
lower bound for the list size of lifted Gabidulin codes for subspace
radius $\tau$.

For the injection distance it holds that
$$d_I(\cR, \rs [\;I_k\; A\;]) = d_R(A_1, A) $$
and if follows right away that the lower bound of the list size of
classical Gabidulin codes for rank radius $t$ is also a lower bound
for the list size of lifted Gabidulin codes for injection radius $t$.

The formula for the list size of classical Gabidulin codes can be
found in~\cite{WZ12}.  \qed
\end{proof}

For the rest of this section let $\tau=2t$, then the two bounds of
Theorem~\ref{trm:list_lower_bound} are equal and asymptotically
become:

\begin{equation}
  \label{eq:lowerBoundList}
  \frac{\Gauss{k}{t}}{q^{(n-k)(\delta-t-1)}}\sim q^{t(k-t)-(n-k)(\delta-t-1)}=q^{-t^2+nt-(n-k)(\delta-1)}.
\end{equation}
(For $t=\delta-1$ this bound becomes $q^{(\delta-1)(k-\delta+1)}$
which does not depend on $n$.)  Similarly to~\cite{WZ12}, one can find
the smallest value of radius $t$, when the exponent
$-t^2+nt-(n-k)(\delta-1)$ appearing in~(\ref{eq:lowerBoundList})
becomes positive.
When this is the case the list variety
  has a positive dimension and the size of the list grows polynomially
  with the field size. The following corollary shows that as a function of $n$
the list size grows exponentially.


 \begin{cor}
   For any $0\leq \epsilon<1$ the list sizes $L_S(2t,n,k,\delta,q)$
   and $L_I(t,n,k,\delta,q)$ are exponential in $n$ if
 $$t\geq (n-\sqrt{n(n-4\delta+4\epsilon)+4k\delta+4k})/2 .$$
\end{cor}

\section{Conclusion and Open Problems}
\label{sec:conclusions}


The balls in $\G$ with a center that is not necessarily in $\G$ are considered with respect to two different distances: the subspace distance and the injection distance. Two different techniques  are used for describing these balls: one is the Pl\"ucker embedding of $\G$ and the second one is a rational parametrization of the matrix representation of the elements in $\G$.
These results can be used for list decoding of constant dimension codes. In particular, we investigate lifted Gabidulin codes and show that these can be described
 by linear equations in either the matrix representation or a
  subset of the Pl\"ucker coordinates.
  The union of these linear equations
   and the linear and bilinear equations which arise from the description of the ball of a given radius in the Grassmannian
 describe the
  list of codewords with distance less than or equal to the given radius from the received word.
 In contrast to the algorithms presented in
\cite{GuXi12,MaVa12} the algorithms presented in this paper work for the complete lifted Gabidulin codes
for any set of parameters $q,n,k,\delta$.

In fact, the theory of Section \ref{sec:LG} holds for any linear rank-metric code, not only Gabidulin codes, hence also the algorithms from Section \ref{sec:list_decoding} work for any lifted linear rank-metric code.

One can easily extend the algorithms presented in this paper
for unions of lifted Gabidulin codes of different length (cf. e.g.
\cite{Ska10,tr13phd}). To do so, one needs to add a preliminary step
in the algorithm where a rank argument decides, which of these lifted Gabidulin
codes can possibly have codewords that are in the ball around the
received word.

The storage needed for both our algorithms is fairly little, the complexity
is polynomial in $n$ but exponential in $k$. Since in applications,
$k$ is quite small while $n$ tends to get large, this is still
reasonable. In future work, we want to improve this complexity
by trying to decrease the size of the system of equations to solve in
the last step of Algorithm \ref{alg:1} on one hand, or to find a better way to solve the system of bilinear equations in Algorithm \ref{alg:2} on the other. 
Moreover, we would like to find other families of
codes that can be described through equations in their Pl\"ucker
coordinates and use this fact to come up with list decoding algorithms
of these other codes.

\section*{Acknowledgment}
The authors wish to thank Antonia Wachter-Zeh
for many helpful discussions. They also thank the anonymous
reviewers for their valuable comments and suggestions that helped to improve the presentation of the paper.



\end{document}